\documentclass[conference]{worldcomp}
\IEEEoverridecommandlockouts
\usepackage{graphicx}
\usepackage{tikz} 
\usepackage{mathdots}
\usepackage{fixltx2e}
 \usepackage{subfig}
 \usepackage{multirow}
\usepackage[hmargin=.75in,vmargin=1in]{geometry}
\usepackage{epsfig,graphicx}
\usepackage{xcolor}
\usepackage{amsfonts}
\usepackage{amsmath}
\usepackage{amssymb}
\usepackage{fixltx2e} 
\usepackage{booktabs}
\usepackage[T1]{fontenc}
\usepackage{times}
\usepackage{caption}
\usepackage{bbm}
\newcommand{\be}{\begin{equation}}
\newcommand{\ee}{\end{equation}}
\newcommand{\ben}{\begin{equation*}}
\newcommand{\een}{\end{equation*}}
\newcommand{\ZZ}{{\mathbbm{Z}}}

\newcommand{\card}{\mathop{\mathrm{card}}}
\newcommand{\f}{{\mathbf f}}
\newcommand{\g}{{\mathbf g}}

\newcommand{\B}{{\mathcal B}}

\newcommand{\T}{{\mathcal T}}
\newcommand{\G}{{\mathcal G}}

\newtheorem{prop}{Proposition}

    \newtheorem{conjecture}{Conjecture}

\begin{document}

\title{\bf Response Curves and Preimage Sequences\\ of Two-Dimensional Cellular Automata}           

\author{
{\bfseries Henryk Fuk\'s and Andrew Skelton}\\
Department of Mathematics, Brock University, St. Catharines, Ontario, Canada.
}

\maketitle

\begin{abstract}
  \noindent 
We consider the problem of finding response curves for a class of binary two-dimensional cellular automata with $L$-shaped neighbourhood. We show that the dependence of the density of ones after an arbitrary number of iterations, on the initial density of ones, can be calculated for a fairly large number of rules by considering preimage sets. We provide several examples and a summary of all known results.
We consider a special case of initial density equal to $0.5$ for other rules and compute explicitly the density of ones  after $n$ iterations of the rule. This analysis includes surjective rules, which in the case
of $L$-shaped neighbourhood are all found to be permutive. We conclude with the observation that all rules for 
which preimage curves can be computed explicitly are either finite or asymptotic emulators of identity or shift.
\end{abstract}

\vspace{1em}

\section{Introduction} 
Cellular automata (CA) can be viewed as computing devices, which take as an
input
some initial configuration. The CA rule is iterated a number of times, 
resulting in a final output configuration.
In many practical problems, e.g., in mathematical modelling, one wants to know
how a CA rule iterated over an initial configuration affects certain aggregate
properties of the configuration, such as, for example, the density of ones. If
we take a randomly generated initial configuration with a given density of ones,
and iterate a given rule $n$ times over this configuration, what is the density
of ones in the resulting configuration? 
We want to know  the ``response curve'', the density of the output as a function
of the density of the input. 

Response curves appear in computational problems, and a classical example of
such a problem in CA theory is the so-called density classification problem
(DCP). If we denote the density of ones in the configuration at time $n$ by
$P_n(1)$, the DCP asks us to find a rule for which $P_{\infty}(1)=1$ if
$P_0(1)>1/2$ and $P_{\infty}(1)=0$ if $P_0(1)<1/2$.
Since it is known that such a rule does not exist \cite{LB95}, once could ask
which response curves are possible in CA rules? We propose to approach this
problem from an opposite direction: given the CA rule, what can we say about its
response curve? It turns out that in surprisingly many cases, the response curve
can be calculated exactly, providing that preimage sets of finite strings under
the CA rule exhibit recognizable patters. 

\section{Definitions} \label{subsub:definitions}
In what follows, we will be concerned with what we call two-dimensional elementary cellular automata, which
have a local function depending on the central site, its right neighbour, and its top neighbour, 
and which allow two states only, 0 and 1. We will say that these are rules with 
 ``L-neighbourhood'', since the neighbourhood has the shape of the letter L.
Such three-input
binary local rules can be considered the simplest ``truly'' two-dimensional CA rules, hence the name ``elementary''.

Before we define such rules formally, we will first introduce the concept of \emph{triangular blocks}, 
defined as regions of 2D configurations in the shape of isosceles right triangles. The set of triangular blocks of size $r$, denoted $\T_r$, is the set consisting of elements
\vspace{-5mm}
\be \label{blockb}
\begin{tabular}{c @{\hspace{1mm}} c @{\hspace{1mm}} c} 
\\ $b_ {1,r}$\\ $\vdots$ & $\ddots$ \\ $b_{1,1}$ & $\hdots$ & $b_{r,1}$,
 \end{tabular} 
\ee
where each $b_{i,j}\in \G$. The set of eight blocks in $T_2$ will be called {\em basic blocks}.

 We may define the {\em local mapping}, or {\em local rule}, of a 2D CA with L-neighbourhood as $g: \T_2 \to \G$. The local mapping $g$ has a corresponding {\em global mapping}, $G: \G^{\ZZ ^2} \to \G^{\ZZ ^2}$ such that $\left(G(s)\right)_{i,j} = g \left(\begin{smallmatrix} s_{i,j+1} \\ s_{i,j} & s_{i+1,j} \end{smallmatrix}\right),$ for any $i,j \in \ZZ$, $s \in \G^{\ZZ ^2}$.

  The block evolution operator $\g: \T_r \to \T_{r-1}$ will be defined as a function which transforms triangular block (\ref{blockb}) into another block, $c \in \T_{r-1}$,
where $c_{i,j}=g \left(\begin{smallmatrix} b_{i,j+1} \\ b_{i,j} & b_{i+1,j} \end{smallmatrix}\right)\in \G$  for 
$i \in \{1,\ldots, r-1 \}$ and $j \in \{1, \ldots, r-i \}$. We denote $\g^n: \T_{r+n} \to \T_r$ to be the operator obtained by composing $\g$ with itself $n$-times.

Occasionally, we will need to define a distance between two configurations. One can show easily that
for $s,t \in \mathcal{G}^{\mathbb{Z}^2}$ and $i,j \in \mathbb{Z}$, the following satisfies all axioms of a metric:
\ben d(s,t) =\hspace{-1mm} \left\{ \hspace{-1mm}
\begin{array}{l l}
  \displaystyle\frac{1}{1+\displaystyle\mathop{\min}_{i,j \in  \mathbb{Z}} \big( \max\{|i|,|j|\} : s_{i,j} \neq t_{i,j})\big)} & \text{if } s \neq t  \\
 0 &  \text{if } s = t\\
\end{array} \right.\hspace{-3mm}.
\een
For 2D CA with L-neighbourhood, we adapt the numbering system used in \cite{wolfram94}. A local rule $g$ is assigned a {\em Wolfram number} W as follows 
\be \label{eqn:Lnumbering}
W(g) = \sum_{a_0,a_1,a_2 \in \{0,1\}} g\left(\begin{smallmatrix} a_0 \\ a_1 & a_2 \end{smallmatrix}\right) 2^{4a_0 + 2a_1 + a_2}.
\ee
We note that, as in the case of radius-1 1D CA, there are 256 possible {\em elementary 2D CA}. Many of these rules are related to each other by the group of 4 transformations $D_1 \times S_2$, where $D_1$ is the dihedral group with a single reflectional symmetry and $S_2$ denotes all permutations of the elements in $\{0,1\}$. Among each class of four (not necessarily distinct) rules, we choose one representative with the smallest Wolfram number. We denote this rule to be a {\em minimal rule}. A list of all 88 minimal rules and their equivalences can be found in \cite{skelton11}.

\section{Densities of Blocks}
We now consider the density response problem.
Suppose that we start with an initial configurations in which a certain proportion of sites 
is in state 1. The simplest way to achieve this is to set each site to be in state 1 with probability $\rho$, 
and 0 with probability $1-\rho$, doing it independently for all sites. This means that the probability
of randomly selected site to be in state 1 is $\rho$. Suppose that we apply $n$ iterates of some  CA rule to such configuration.
What is the probability that in the resulting configuration, the state of a randomly selected site is 1? 

In order to formulate this problem more precisely, we will use the concept of probability measure, similarly as
done in \cite{paper39}, for one-dimensional CA.

Given a block $b \in \T_r$, we define a {\em cylinder set given by $b$}, $C_{i,j}(b)$, as the set of all configurations in which block $b$ is fixed and placed at coordinate $(i,j)$ aligned at the lower-left element of $b$. We define a {\em measure} of such as cylinder set, $\mu\left[ C_{i,j}(b)\right]$, to be the probability of occurrence of block $b$ placed as above. If the measure is translationally invariant we may drop the indices $i,j$. For $\rho \in [0, 1]$, the Bernoulli measure
is a measure where all sites are independently set to 1 with probability $\rho$, and to 0 with
probability $1-\rho$. In such case, we have
\be \label{eqn:bernoulli}\mu_\rho \left[C(b)\right] = \rho^{j}(1-\rho)^{(r^2+r)/2-j},\ee
where $j$ is a number of cells in state 1 in $b$. 

We now consider the action of the global mapping $G$ on the measure of a cylinder set given by block $b$, which yields
\be \left(G\mu_\rho\right)\left[C(b)\right] = \mu_\rho\left[G^{-1}\left(C(b)\right)\right].\ee
Considering instead $n$ iterations of $G$, we obtain
\be \label{eqn:pnb1}\left(G^n\mu_\rho\right)\left[C(b)\right] = \mu_\rho\left[G^{-n}\left(C(b)\right)\right].\ee
If we let $\g^{-n}(b)$ be the set of all {\em $n$-step preimages of block $b$}, that is, the set of all blocks $a$ such that $\g^{n}(a) = b$, then we can write
\be \label{eqn:rhs1}
\mu_\rho\left[G^{-n}\left(C(b)\right)\right] = \sum_{a \in \g^{-n}(b)} \mu_\rho \left[ a \right].
\ee
Using the notation $P_n(b) = \left(G^n\mu_\rho\right)\left[C(b)\right]$, we write (\ref{eqn:pnb1}) as
\be \label{eqn:pn}
P_n(b) = \sum_{a \in \g^{-n}(b)} P_0(a).
\ee
If $b=1$, and if the initial measure is Bernoulli, then in the above formula each $P_0(a)$ depends
only on $\rho$, where $\rho=P_0(1)$. $P_n(1)$ can then be interpreted as the density of 1s in the
configuration obtained by iterating the CA rule $n$ times starting from disordered 
initial configurations with density of ones equal to $\rho$.

Plot of $P_n(1)$ versus $\rho$ will be called a {\em response curve} for each elementary 2D CA. In the special case when $\rho=1/2$, the probability of any block of a given size is equally likely and (\ref{eqn:pn}) can be expressed as
\be \label{eqn:phalf}
P_n(b) = 2^{-(r+n+1)(r+n)/2} \;\text{card}\left[\g^{-n}(b)\right],
\ee
where $\text{card}\left[\g^{-n}(b)\right]$ denotes the number of elements in the set $\g^{-n}(b)$.
 If we want to indicate that we consider the special case of  $\rho=1/2$, we will use the notation $P_n^{(s)}(b)$, at the sequence of $P_n^{(s)}(b)$ for $n=0,1,2\ldots$ will be called  a {\em response sequence}. Finally, we denote $P(b)$ to be the {\em asymptotic density of block $b$}, which we obtain by taking the limit of $P_n(b)$ as $n \to \infty$ (if the limit exists).


\section{Theoretical Response Curves}
For 26 minimal rules, we were able to determine an explicit response curve formula. In some cases, we found that the response curve was independent of $n$. In other cases, the response curve was dependent on $n$, and then a separate formula for the asymptotic density could be obtained. We present in detail three examples of each types. In each example, we describe the structure of the preimage sets but, due to space constraints, we omit direct proofs while noting that each case can be proved easily by induction.

\begin{figure*}
  \begin{center}
  \subfloat[Rule 2]{\includegraphics[scale=0.25]{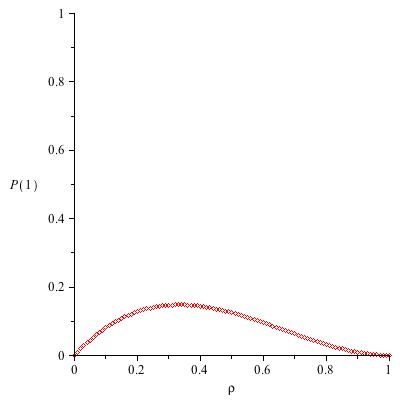}\label{fig:2}}         
  \subfloat[Rule 3]{\includegraphics[scale=0.25]{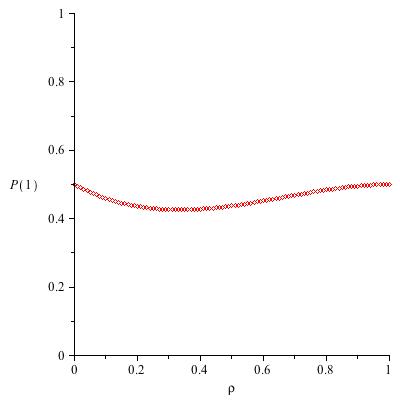}\label{fig:3}}
  \subfloat[Rule 10]{\includegraphics[scale=0.25]{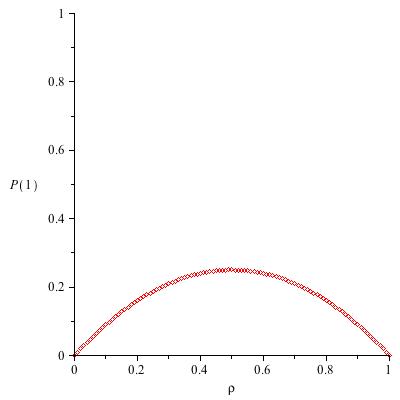}\label{fig:10}}
  \subfloat[Rule 42]{\includegraphics[scale=0.25]{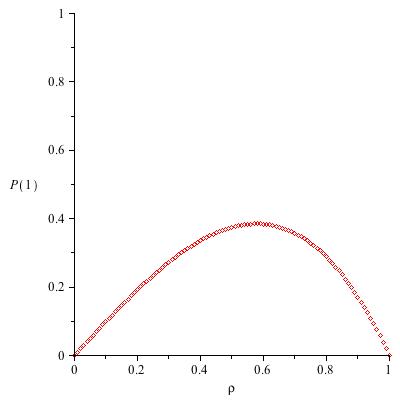}\label{fig:42}}  
  \subfloat[Rule 138]{\includegraphics[scale=0.25]{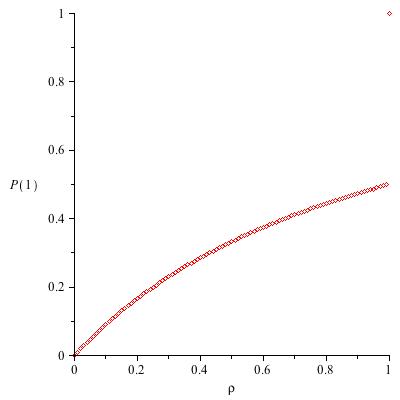}\label{fig:138}}          
  \caption{Experimental Response Curves}
  \label{fig:experimentalcurves2}
  \end{center} 
  \vspace{-3mm}
\end{figure*}

\subsection{Rules with Constant Density}
In each of the following examples the formula for the response curve has no dependence on $n$. Therefore, the formula for the asymptotic density is the same as the response curve. We provide detailed analysis for Rules 0, 3 and 42 and the remaining results are presented in Table \ref{table:constant}.
\begin{prop}
The response curve for Rule 0 is $P_n(1) = 0$.
\end{prop}
\begin{proof}
There are no triangular blocks of any size that can be mapped under $\g^n_{0}$ to single block 1. Therefore, $\text{card}\left[\g_{0}^{-n}(b)\right] = 0$ and we apply (\ref{eqn:pn}) to obtain our result.
\end{proof}

\begin{prop}
The response curve for Rule 42 is $$P_n(1) = \rho(1-\rho)(1+\rho).$$
\end{prop}
\begin{proof}
It can be shown by induction that the only blocks that map to a single 1 under $\g^n_{42}$ are either blocks in $\T_n$ where $b_{n,1}=1,b_{n-1,2}=0$ and all other elements are arbitrary, or blocks in $\T_n$ where $b_{n,1}=b_{n-1,2}=1, b_{n-1,1}=0$ and all other elements are arbitrary.
Using (\ref{eqn:bernoulli}) 
and (\ref{eqn:pn}), we conclude that $P_n(1) = \rho(1-\rho) + \rho^2 (1-\rho)$, which simplifies to the desired result. An experimental curve confirming this result is presented in Figure \ref{fig:42}. 
\end{proof}

\begin{prop}
The response curve for Rule 3 is 
\ben
P_n(1) = \begin{cases} (1-\rho)^2 & \quad \text{if $n$ even,} \\
\rho(1+\rho-\rho^2) & \quad \text{if $n$ odd.} \end{cases}
\een
\end{prop}
\begin{proof}
Since Rule 3 has period-2 behaviour, we must consider cases when $n$ is odd and when $n$ is even. When $n$ is odd, the only blocks that map to a single 1 under $n$-iterations of $\g_{3}$ are blocks in $\T_n$ where $b_{1,(n+1)/2}=b_{1,(n+3)/2}=0$ and all other elements are arbitrary.

%
When $n$ is odd, the only blocks that map to 1 under $\g^n_{3}$ are blocks in $\T_n$ where $b_{1,(n+2)/2}=0$ and all other elements are arbitrary, or blocks in $\T_n$ where $b_{1,n/2}=b_{1,(n+4)/2}=1, b_{1,(n+2)/2}=0$ and all other elements are arbitrary.




 Using (\ref{eqn:bernoulli}) we can determine the initial probability of occurrence of blocks of either type and using (\ref{eqn:pn}) obtain a formula which simplifies to the desired result. 
Note that in the special case when $\rho = 1/2$, we conclude that $P_n^{(s)}(1)$ equals $1/4$ when $n$ is even, and equals $5/8$ when $n$ is odd.
We averaged our experimental results over an even number of time steps, thus Figure \ref{fig:3} has the form
 \begin{align*}
 P_n(1) &= \frac{1}{2} \left[ (1-\rho)^2 + \rho(1+\rho-\rho^2)\right].
 \end{align*} 
\vspace{-4mm}
\end{proof}
\begin{table}
\centering
\caption{Constant Density Rules}\label{table:constant}
\begin{tabular}{c || c | c}
Rules & $P_n(1)$ & $P_n^{(s)}(1)$ \\ \hline\hline
0 & 0 & 0\\ \hline
2 & $\rho(1-\rho)^2$ & $1/8$ \\ \hline
\multirow{2}{*}{3} & $(1-\rho)^2$ & $5/8$ ($n$ odd) \\ & $\rho(1+\rho-\rho^2)$ & $1/4$ ($n$ even)  \\ \hline
4 & $\rho(1-\rho)^2$ & $1/8$ \\ \hline
\multirow{2}{*}{5} & $(1-\rho)^2$ & $5/8$ ($n$ odd) \\ & $\rho(1+\rho-\rho^2)$ & $1/4$ ($n$ even)  \\ \hline
10 & $\rho(1-\rho)$ & $1/4$ \\ \hline
12 & $\rho(1-\rho)$ & $1/4$ \\ \hline
\multirow{2}{*}{51} & $1-\rho$ \quad($n$ odd) & \multirow{2}{*}{$1/2$} \\ & $\rho$ \quad($n$ even) &  \\ \hline 
34 & $\rho(1-\rho)$ & $1/4$ \\ \hline
42 & $\rho(1-\rho)(1+\rho)$ & $3/8$ \\ \hline
\multirow{2}{*}{51} & $1-\rho$ \quad($n$ odd) & \multirow{2}{*}{$1/2$} \\ & $\rho$ \quad($n$ even) &  \\ \hline 
76 & $\rho(1-\rho)(1+\rho)$ & $3/8$ \\\hline
170 & $\rho$ & $1/2$\\ \hline
204 & $\rho$ & $1/2$\\
\end{tabular}
 \vspace{-5mm}
\end{table}

\subsection{Rules with Decaying Density}
In each of the following examples the formula for the response curve is dependent on $n$, and thus we can also determine an asymptotic density formula. We provide detailed analysis for Rules 32, 128 and 138, while the remaining results are presented in Table \ref{table:decay}.

\begin{prop}
The response curve for Rule 128 is $$P_n(1) = \rho^{(n^2+3n+2)/2}.$$
\end{prop}
\begin{proof}
The only block mapping to a single 1 under $\g^n_{128}$ is the block consisting entirely of ones. We use (\ref{eqn:bernoulli}) to find the initial probability of this block and (\ref{eqn:pn}) produces our result.  
Thus, the asymptotic density under Rule 128 is 
\ben
P(1) = \lim_{n \to \infty} P_n(1) = \begin{cases} 0 & \quad \text{if $\rho \neq 1$,} \\
1 & \quad \text{if $\rho =1$.} \end{cases}
\een

\end{proof}

\begin{prop}
The response curve for Rule 32 is $$P_n(1) = \rho^{n+1} (1-\rho)^{n}.$$
\end{prop}
\begin{proof}
Under rule 32, the only blocks that map to a single 1 under $\g^n_{32}$ are of the form
\vspace{-1.5mm}
 \ben
\begin{tikzpicture} 
\draw (0,0) node{$\star$};
\draw (0,0.5) node{$\vdots$};
\draw (0,0.8) node{$\star$};
\draw (0,1.15) node{$0$};
\draw (0,1.5) node{$1$};
\draw (0,1.15) node{$0$};
\draw (0.5,0) node{$\hdots$};
\draw (0.5,0.5) node{$\ddots$};
\draw (1,0) node{$\star$};
\draw (1.3,0) node{$0$};
\draw (1.6,0) node{$1$};
\draw (0.95,0.55) node{$\ddots$};
\draw (0.5,0.9) node{$\ddots$};
\draw (0.55,1.25) node{$\ddots$};
\draw (1.2,0.65) node{$\ddots$};

\draw[color=black!100!white]  (-0.3,-0.05) -- (-0.4,-0.05) -- (-0.4,0.85) -- (-0.3,0.85);
\draw (2,0) node{.};
\draw (-0.95,0.4) node{$n-1$};
\end{tikzpicture}
\vspace{-2.5mm}
\een

Using (\ref{eqn:bernoulli}), we can determine the initial probability of occurrence of blocks of this type and using (\ref{eqn:pn}), we obtain our result. The asymptotic density is thus $P(1) = 0.$ 

\end{proof}

\begin{prop}
The response curve for Rule 138 is $$P_n(1) =  \frac{\rho^{2n+2}+\rho}{\rho+1}.$$
\end{prop}
\begin{proof}
The only blocks that map to a single 1 under $n$-iterations of $\g_{138}$ are comprised entirely of arbitrary elements in the top $n-1$ rows, and have their lower two rows, where $i$ ranges from $1$ to $n+1$, of the form
\vspace{-1mm}
\ben
\begin{tikzpicture} 
\draw (0,-0.2) node{$\star$};
\draw (0.4,-0.2) node{$\hdots$};
\draw (0.8,-0.2) node{$\star$};
\draw (1.05,-0.2) node{$\star$};
\draw (1.3,-0.2) node{$1$};
\draw (1.75,-0.2) node{$\hdots$};
\draw (2.15,-0.2) node{$1$};
\draw (2.4,-0.2) node{$1$};
\draw (2.65,-0.2) node{$1$};

\draw (0,0.2) node{$\star$};
\draw (0.4,0.2) node{$\hdots$};
\draw (0.8,0.2) node{$\star$};
\draw (1.05,0.2) node{$0$};
\draw (1.3,0.2) node{$1$};
\draw (1.75,0.2) node{$\hdots$};
\draw (2.15,0.2) node{$1$};
\draw (2.4,0.2) node{$1$};

\draw[color=black!100!white]  (1.2,-0.4) -- (1.2,-0.5) -- (2.75,-0.5) -- (2.75, -0.4);
\draw (3,-0.5) node{,};
\draw (2,-0.75) node{$i$};
\end{tikzpicture} \vspace{-3mm}
\een
Using (\ref{eqn:bernoulli}) we can determine the initial probability of occurrence of blocks for each possible value of $i$. Summing over all $i$ and using (\ref{eqn:pn}), we conclude that 
\begin{align*}
P_n(1) &= \rho^{2n+1} + \sum_{i=1}^n \rho^{2i-1} (1-\rho) \\
&= \rho^{2n+1} + \frac{1-\rho}{\rho} \left( \frac{\rho^2 \left(\rho^{2n}-1\right)}{\rho^2-1}\right), 
\end{align*}
which simplifies to our desired result.
\end{proof}

Again, we can find the asymptotic density as
\ben
P(1) = \lim_{n \to \infty} P_n(1) = \begin{cases} \frac{\rho}{1+\rho} & \quad \text{if $\rho \neq 1$,} \\
1 & \quad \text{if $\rho =1$.} \end{cases}
\een
This result is confirmed by the experimental curve in Figure \ref{fig:138}. Note that while the response curve is continuous, the asymptotic density has a discontinuity at $\rho=1$, corresponding to an initial condition consisting entirely of ones.

\begin{table}
\centering \caption{Density Decaying Rules}  \label{table:decay}
\begin{tabular}{c || c | c}
Rules & $P_n(1)$ & $P(1)$ \\ \hline\hline
8 & $\rho^{n+1}(1-\rho)^n$ & 0\\ \hline
32 & $\rho^{n+1}(1-\rho)^n$ & 0\\ \hline
40 &  $2^n\rho^{n+1}(1-\rho)^n$ & 0\\ \hline
72 & $2^n\rho^{n+1}(1-\rho)^n$ & 0\\ \hline
\multirow{2}{*}{128} & \multirow{2}{*}{$\rho^{(n^2+3n+2)/2}$} & $0$ \quad if $\rho \neq 1$ \\ & & 1 \quad if $\rho = 1$  \\ \hline
130 & \multicolumn{2}{c}{see $\cite{skelton10}$ for complete analysis}  \\ \hline
132 & \multicolumn{2}{c}{can be derived from Rule 130}  \\ \hline
\multirow{2}{*}{136} & \multirow{2}{*}{$\rho^{n+1}$} & $0$ \quad if $\rho \neq 1$ \\ & & 1 \quad if $\rho = 1$  \\ \hline
\multirow{2}{*}{138} & \multirow{2}{*}{$\frac{\rho^{2n+2}+\rho}{\rho+1} $} & $\frac{\rho}{1+\rho}$ \quad if $\rho \neq 1$ \\ & & 1 \quad if $\rho = 1$  \\ \hline
\multirow{2}{*}{140} & \multirow{2}{*}{$\frac{\rho^{2n+2}+\rho}{\rho+1} $} & $\frac{\rho}{1+\rho}$ \quad if $\rho \neq 1$ \\ & & 1 \quad if $\rho = 1$  \\ \hline
\multirow{2}{*}{160} & \multirow{2}{*}{$\rho^{n+1}$} & $0$ \quad if $\rho \neq 1$ \\ & & 1 \quad if $\rho = 1$  \\ \hline
\multirow{2}{*}{162} & \multirow{2}{*}{$\frac{\rho^{2n+2}+\rho}{\rho+1} $} & $\frac{\rho}{1+\rho}$ \quad if $\rho \neq 1$ \\ & & 1 \quad if $\rho = 1$  \\
\end{tabular}
 \vspace{-4mm}
\end{table}

\section{Theoretical Response Sequences}
In some cases we were unable to determine an explicit expression for the response curve of a given rule, but we were able to derive an explicit formula for $\text{card}\left[\g^{-n}(1)\right]$, and thus use (\ref{eqn:phalf}) to obtain a response sequence. For 21 additional rules, we were able to either prove or conjecture a response sequence. We first consider the class of surjective rules.

\begin{figure*}
  \begin{center}
  \subfloat[Rule 2]{\includegraphics[scale=0.25]{rule2.jpeg}}         
  \subfloat[Rule 108]{\includegraphics[scale=0.25]{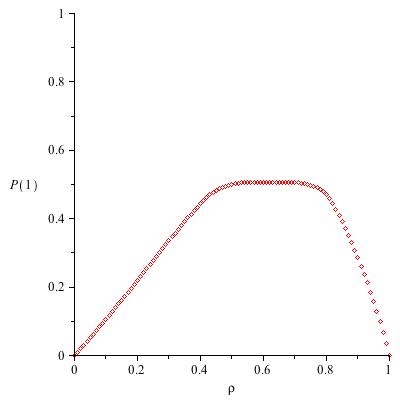}}
  \subfloat[Rule 154]{\includegraphics[scale=0.25]{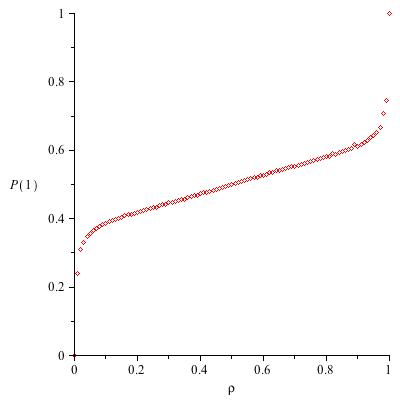}}
  \subfloat[Rule 168]{\includegraphics[scale=0.25]{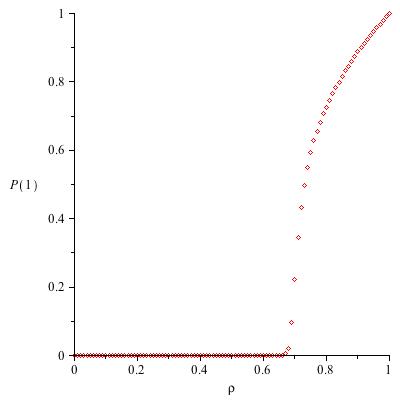}}
  \subfloat[Rule 172]{\includegraphics[scale=0.25]{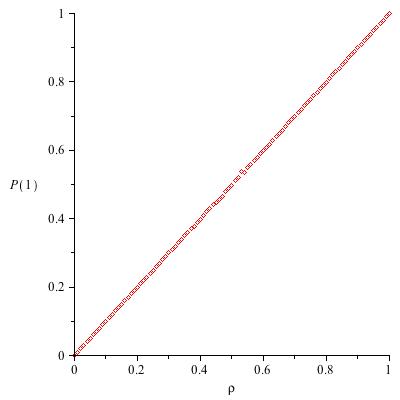}}           
  \caption{Experimental Response Curves}
  \label{fig:experimentalcurves}
  \end{center}
   \vspace{-3mm}
\end{figure*}

\subsection{Surjective Rules}
 Sites belonging to the L-shaped neighbourhood $(\begin{smallmatrix} a_{0,1} \\ a_{0,0} & a_{1,0} \end{smallmatrix})$ will be identified by their indices as $(0,1)$, $(0,0)$, and $(1,0)$. Similarly as in \cite{hedlund69}, a local function $g$ will be called \textit{permutive} with respect to the $(0,1)$ site if for any choice of $y,z \in \mathcal{G} $ the function $x \to g(\begin{smallmatrix} x \\ y & z \end{smallmatrix})$ is one-to-one. Permutivity with respect to the central site  $(0,0)$ or the right neighbour $(0,1)$ is defined similarly.  We now find a response sequence for rules permutive with respect to site $(0,0)$.

\begin{prop} \label{prop:centre}
The response sequence for Rules 15, 30, 45, 51, 54, 57, 60, 90, 105, 106, 108, 150, 154, 156, 170 and 204 is $P_n^{(s)}(1) = 1/2$.
\end{prop}
\begin{proof}
There are 16 rules permutive with respect to the centre site, of which the following 9 are minimal: 51, 54, 57, 60, 105, 108, 150, 156 and 204. If a rule is permutive with respect to $(0,0)$, then there must exist
numbers $x_0, \ldots, x_3 \in\{0,1\}$ such that the local function takes the form
 \be\label{Lshaped-form}
g(\begin{smallmatrix} a_0 \\ a_1 & a_2 \end{smallmatrix}) = \hspace{-1mm}\left\{ \hspace{-1mm}
\begin{array}{l l}
  0 & \text{if}\;\; (\begin{smallmatrix} a_0 \\ a_1 & a_2 \end{smallmatrix}) \in \{ \begin{smallmatrix} \textbf{0} \\ x_0 & \textbf{0} \end{smallmatrix},
\begin{smallmatrix} \textbf{0} \\ x_1 & \textbf{1} \end{smallmatrix}, 
\begin{smallmatrix} \textbf{1} \\ x_2 & \textbf{0} \end{smallmatrix},
\begin{smallmatrix} \textbf{1} \\ x_3 & \textbf{1} \end{smallmatrix}\}\\\\
  1 & \text{if}\;\; (\begin{smallmatrix} a_0 \\ a_1 & a_2 \end{smallmatrix}) \in \{ \begin{smallmatrix} \textbf{0} \\ \overline{x}_0 & \textbf{0} \end{smallmatrix},
\begin{smallmatrix} \textbf{0} \\ \overline{x}_1 & \textbf{1} \end{smallmatrix},
\begin{smallmatrix} \textbf{1} \\ \overline{x}_2 & \textbf{0} \end{smallmatrix},
\begin{smallmatrix} \textbf{1} \\ \overline{x}_3 & \textbf{1} \end{smallmatrix}\}\\
\end{array} \right.\hspace{-3mm},
\ee
where $\overline{x}_i$ denotes $1-x_i$. Assuming the above form of $g$, 
let us consider an arbitrary block $b \in \mathcal{T}^n$.
We will now show how to construct all preimages of $b$ under $\mathbf{g}$. First of all, we
 claim that blocks $c \in \mathcal{T}^{n+1}$ of the form
\begin{equation} c= \begin{tabular}{c @{\hspace{2mm}} c @{\hspace{2mm}} c @{\hspace{+2mm}} c @{\hspace{+2mm}} c @{\hspace{+2mm}} c} 
$\alpha_1$
 \\ $c_{1,n}$ & $\ddots$
 \\ $\vdots$ & $\ddots$  & $\ddots$
 \\ $c_{1,1}$ & $\hdots$ &$c_{n,1}$  & $\alpha_{n+1}$
 \end{tabular}.
 \end{equation}
are the only preimages of~$b$,  where each $\alpha_i$  ($1 \leq i \leq n+1$) is an arbitrary value in $\{0,1\}$, and values of $c_{i,j} \in \{0,1\}$ can be 
 determined by an iterative algorithm.

To see that this is indeed true, we now present an algorithm with which we can construct all possible preimages:
 \begin{enumerate}
 \item Starting from $b_{1,n}$, we wish to find all neighbourhoods $\{ \begin{smallmatrix}   a_0 \\ a_1 & a_2 \end{smallmatrix} \}$ such that $g\big(\; \begin{smallmatrix} a_0 \\ a_1 & a_2 \end{smallmatrix}\;\big ) = \begin{smallmatrix} \\ b_{1,n} & \end{smallmatrix}$. The structure of the local mapping gives us four possible such neighbourhoods $\{ \begin{smallmatrix} a_0 \\ a_1 & a_2 \end{smallmatrix}\} = \{\begin{smallmatrix} \alpha_0 \\ c_{1,n} & \alpha_1 \end{smallmatrix}\}$, where 
 \ben
 c_{1,n}=(1-b_{1,n}) x_{2\alpha_1+\alpha_2} + b_{1,n}(1-x_{2\alpha_1+\alpha_2}),
 \een
  and the  values of $\alpha_1$ and $\alpha_2$ are arbitrarily selected.

We now repeat step 2 for all values of $i \in \{2,\dots,n\}$.

  \item Since $b_{i,n-i+1}$ is given and $\alpha_i$ has been freely chosen in the previous iteration, we wish to know all neighbourhoods $\{ \begin{smallmatrix}   a_0 \\ a_1 & a_2 \end{smallmatrix} \}$, such that $g\big(\; \begin{smallmatrix} a_0 \\ a_1 & a_2 \end{smallmatrix}\;\big ) = \begin{smallmatrix} \\ b_{i,n-i+1} & \end{smallmatrix}$. The structure of the local mapping gives two possible neighbourhoods $\{ \begin{smallmatrix} a_0 \\ a_1 & a_2 \end{smallmatrix}\} = \{\begin{smallmatrix} \alpha_i \\ c_{i,n-i+1} & \alpha_{i+1} \end{smallmatrix}\}$, where 
 \begin{align*}
 c_{i,n-i+1}&=(1-b_{i,n-i+1}) x_{2\alpha_i+\alpha_{i+1}} +\\
 &\quad+ b_{i,n-i+1}(1-x_{2\alpha_i+\alpha_{i+1}}),
 \end{align*}
  and $\alpha_{i+1}$ is another arbitrarily selected value.


We now construct the rest of the preimage and show that all other values are uniquely determined based on each choice of the $\alpha$ values  in the top diagonal.
For all values of $j \in \{1,\dots,n-i\}$ and then for all $i \in \{1,\dots,n-j\}$, we repeat step 3 as follows.
  
  \item Since $b_{i,n-i-j+1}$ is fixed, we wish to know all neighbourhoods $\{ \begin{smallmatrix}   a_0 \\ a_1 & a_2 \end{smallmatrix} \}$, such that $g\big(\; \begin{smallmatrix} a_0 \\ a_1 & a_2 \end{smallmatrix}\;\big ) =  \begin{smallmatrix} \\b_{i,n-i-j+1} & \end{smallmatrix}$. Since $c_{i,n-i-j+2}$ and $c_{i+1,n-i-j+1}$ were fixed in a previous iteration, the structure of the local mapping tells us that our neighbourhood must have the form $\{ \begin{smallmatrix}   a_0 \\ a_1 & a_2 \end{smallmatrix} \}=\{ \begin{smallmatrix}  c_{i,n-i-j+2}  \\ c_{i,n-i-j+1} & c_{i+1,n-i-j+1} \end{smallmatrix} \}$, where
 \ben
 c_{i,n-i-j+1}=(1-b_{i,n-i-j+1}) x_{i'} + b_{i,n-i-j+1}(1-x_{i'}),
 \een
and  $i' =2c_{i,n-i-j+2}+c_{i+1,n-i-j+1}$. Note that no new arbitrary parameter appears here, thus the neighbourhood is determined uniquely.
  \end{enumerate}
  
The only arbitrary values in the preimage are the $(n+1)$ values of $\alpha_i$ on the main diagonal. Therefore, we know that there are exactly $2^{n+1}$ preimages for a given  $b \in \mathcal{T}^{n}$. Therefore, we can see that $\text{card}\left[\g^{-n}(1) \right] = 2^{(n^2+3n+2)/2}$. Now, using (\ref{eqn:phalf}), we conclude that $P_n^{(s)}(1)=1/2$ for all $n$. Considering rules permutive with respect to the other two sites, we conclude that also Rules 15, 30, 45, 90, 106, 154 and 170  possess a response sequence $P_n^{(s)}(1)=1/2$. 
\end{proof}

It turns out that these rules are the class of surjective 2D CA with L-neighbourhoods. In one dimension, it is known that rules permutive with respect to one of the variables located at the left or the right end of the neighbourhood are surjective, as proved in \cite{hedlund69}. Recently, this result has been generalized to two dimensions by Dennunzio and Formenti \cite{dennunzio08}, who demonstrated that any rule with Moore neighbourhood (of any radius) which is permutive with respect to one of the corner sites is surjective. We now show how one can prove a similar result specifically for the L-shaped neighbourhood, adapting the idea in \cite{sutner99} to 2D CA.

\begin{prop}
If the local mapping of an elementary 2D CA with L-neighbourhood is permutive with respect to any site, then the corresponding global mapping is surjective.
\end{prop}
\begin{proof}
From Proposition \ref{prop:centre}, we know that for any permutive rule and all $b \in \mathcal{T}^n, (n \geq 1)$, $\text{card}[\mathbf{g}^{-1}(b)] = 2^{n+1}$. Consider any infinite configuration, $t \in \mathcal{G}^{\mathbb{Z}^2}$. Define for all $n \geq 1$, the set,
$
S_n = \{ s \in \mathcal{G}^{\mathbb{Z}^2} : \mathbf{g}(s_{[n+1]}) = t_{[n]}\},
$
where $s_{[n+1]}$ denotes a block of size $n$ contained in an infinite configuration $s \in \mathcal{G}^{\mathcal{Z}^2}$ and placed at $(0,0)$. Our assumption guarantees that all $S_n$ are non-empty for $n \geq 1$. We also know that $S_{n+1} \subseteq S_n$. 
We consider the complement of $S_n$, the set $\overline{S_n} = \{ s \in \mathcal{G}^{\mathbb{Z}^2} : \mathbf{g}(s_{[n+1]}) \neq t_{[n]}\}$, to show that $S_n$ is a clopen set. 

We first show that $\overline{S_n}$ is open. Let $s \in \overline{S_n} \subset \{0,1\}^{\mathbb{Z}^2}$ be an arbitrary configuration. For all $\epsilon > 0$, we choose $k \in \mathbb{Z}$, where $k > n$, such that $\frac{1}{k+1} < \epsilon$. We now pick an infinite configuration $s' \in \{0,1\}^{\mathbb{Z}^2}$ such that $d(s,s') = \frac{1}{k^\ast + 1}$, where $k^\ast > k$. Since $s \in \overline{S_n}$, we know that $s' \in \overline{S_n}$, and
\ben 
d(s,s') = \frac{1}{k^\ast + 1} < \frac{1}{k+1} < \epsilon.
\een
Thus, $\overline{S_n}$ is open. Similar analysis shows that $S_n$ must also be open, and thus $S_n$ is a clopen set. By the Nested Set Theorem \cite{marsden93}, there must exist $s \in \mathcal{G}^{\mathbb{Z}^2}$, such that $F(s)=t$. 
\end{proof}

To conclude that these are the only surjective rules, we use the reverse direction of the Balance Theorem.

  \begin{prop} \label{prop:balance}
  If a elementary 2D CA with L-neighbourhood is surjective, then for all $n \geq 1$ and all blocks $b \in \mathcal{T}^n$, 
$
\text{card}[\mathbf{g}^{-1}(b)] = 2^{n+1}.
$
\end{prop}
\begin{proof}
The Balance Theorem was proved in 1D in \cite{hedlund69} and in 2D in \cite{maruoka76}. A version of the proof specifically tailored for the L-neighbourhood is to be reported elsewhere \cite{skelton11}.
\end{proof}

For all other elementary rules, we performed a computerized search and found blocks for which no preimages exist.
By Proposition \ref{prop:balance}, these rules must be non-surjective.

\subsection{Conjectured Response Sequences}
To find response sequences for the remaining rules, we performed an exhaustive search through all potential preimages for each rule. For the L-neighbourhood, the number of potential preimages is $2^{(n^2+3n+2)/2}$, which makes searches for large $n$ impossible. We performed our searches using the Shared Hierarchical Academic Research Computing Network (SHARCNET) and we were able to obtain cardinalities of preimage sets to level $n=7$. We then attempted to conjecture a formula for the sequence using the first six terms, and checked the conjecture with the seventh term. 

Rules 23, 27, 29, 43, 46, 58, 77, 78, 142, 172, 178, 184  each shared the first seven terms of the preimage sequence with the surjective rules above, so that for these rules we conjecture that $P_n^{(s)}(1)=1/2$.
For all remaining rules, a list of the first seven preimage cardinalities is available upon request. 

\section{Experimental Response Curves}
For those rules for which an explicit response curve formula could not be derived, we were able to perform computer simulations to obtain experimental response curves. We start with a square configuration of $250000$ elements and we iterate $1000/\rho(1-\rho)$ times when $\rho \in (0,1)$ and $100000$ times otherwise,
with periodic boundary conditions and averaging density over the last 10 time steps
and over 10 iterations from different initial conditions. Examples of such experimental curves are presented in Figures \ref{fig:experimentalcurves2} and \ref{fig:experimentalcurves}. We note in passing that one of the examples shown in Figure \ref{fig:experimentalcurves}, namely Rule 168, exhibits response curve resembling ``phase transition'', that is,
discontinuity of the derivative. None of  1D elementary rules exhibits such behaviour.

\section{Rule emulation}
We will now briefly turn our attention to dynamics of 2D rules. When one prints sample spatiotemporal
diagrams of 2D rules with L-shaped neighbourhood (not shown here for the lack of space), one can
easily observe that all rules for which density response curves can be calculated theoretically exhibit
somewhat ``simple'' dynamics. A convenient was to describe this ``simplicity'' is to say that after 
a few iterations these rules essentially behave like identity or shift. In order to formalize this statement, we need
to introduce the concept of emulation, first finite and then asymptotic.
\subsection{Finite Rule Emulation}
We say that {\em Rule X emulates Rule Y at level $n$} if, 
\be 
\g^{n+1}_{X} (b)= g_{Y}\left(\g^{n}_{X} (b) \right).
\ee
for any block $b \in \B^{n+2}$. We will demonstrate this with an example. Consider Rule 76, with  a local rule given by
\ben
g_{76} \left(\begin{smallmatrix} x &\\ y & z \\ \end{smallmatrix}\right) = (1-x)y(1-z) + (1-x)yz + xy(1-z) = y(1-xz).
\een
We now compose $g_{76}$ with itself as follows
\begin{align*}
\g_{76}^2 (b) &= g_{76} \left( \g_{76} \left( \begin{smallmatrix} x_0 \\ x_1 & x_2 \\ x_3 & x_4 & x_5 \end{smallmatrix} \right) \right) \\
&= x_3(1-x_1x_4)\left( 1- x_1(1-x_0x_2)x_4(1-x_2x_5) \right) \\
&= x_3(1-x_1x_4) = g_{204} \left(\g_{76}(b)\right),
\end{align*}
where we have used the fact that when $x \in \{0,1\}$, we know that $x^2=x$. We therefore conclude that Rule 76 emulates identity at level 1. We checked all $88 \times 87$ pairs of distinct elementary rules for finite rule emulation. In Figure \ref{figure:emulation}, we show all level 1 emulation relations between all minimal elementary 2D rules with L-neighbourhood as directed graphs in which an arrow travels from X to Y if and only if Rule X emulates Rule Y at level 1. In Figure \ref{fig:identity} are all rules which finitely emulate the identity Rule 204. In Figure \ref{fig:shift} are all rules which finitely emulate the left shift Rule 170. Finally, in Figure \ref{fig:other} are another class of interrelated emulation rules. In addition to the rules in the graph, we also discovered that rules 6, 14, 18 and 50 emulate rules 134, 142, 146 and 178 respectively.

\begin{figure*}
  \centering
  \subfloat[Rules Emulating Identity]{\label{fig:identity}\includegraphics[scale=0.35]{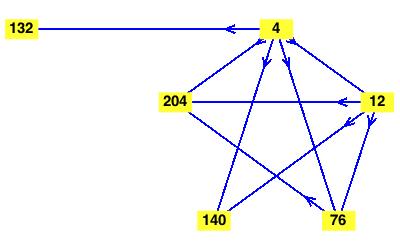}}             
  \subfloat[Rules Emulating Shift]{\label{fig:shift}\includegraphics[scale=0.35]{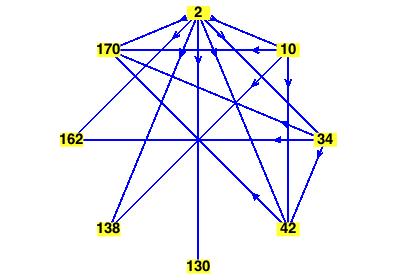}}
  \subfloat[Other Finite Emulations]{\label{fig:other}\includegraphics[scale=0.25]{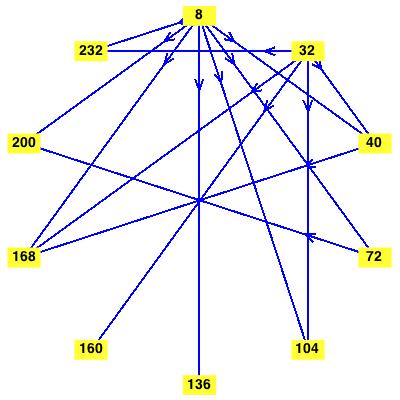}}
  \caption{Finite Emulation Relations}
  \label{figure:emulation} 
\end{figure*} 

\subsection{Asymptotic Rule Emulation}
In \cite{fuks02}, the author defined the following metric to describe the distance between two elementary 1D cellular automata rules. We adapt this and define the following metric to describe the distance between two elementary 2D cellular automata rules with L-neighbourhood
\be
d(f,g) = 2^{(-k^2-3k-2)/2} \sum_{b \in \B_k} \left| f(b) - g(b) \right|.
\ee
We say that {\em Rule f asymptotically emulates Rule g} if
\be 
\lim_{n \to \infty} d \left( f^{n+1}, f^{n} \circ g \right) = 0.
\ee
We now derive a useful equation with which we can calculate the distance between two rules at a given level-$n$. First, we define the following function for any block $b \in \B$,
\ben
\left( f \oplus g \right)(b) = f(b) + g(b) \mod 2,
\een
which outputs 1 if and only if $f(b) \neq g(b)$. Thus, we can use this function to count the number of blocks on which local mappings $f$ and $g$ differ. Adapting Proposition 3 from \cite{fuks02}, we obtain the following proposition (proof in \cite{fuks02}).

\begin{prop} \label{prop:distance}
If $f,g$ are 2D local L-neighbourhood mappings, $A_0 = \left( f \oplus g \right)^{-1}(1)$, and $A_n = \f^{-n}(A_0)$, then
\be \label{eqn:distancen}
 d \left( f^{n+1}, f^{n} \circ g \right) = \frac{\card \left[ A_n \right]}{2^{(n^2+5n+6)/2}}.
\ee
\end{prop}

We demonstrate this procedure with an example.

\begin{prop}
2D CA Rule 160 asymptotically emulates the identity rule.
\end{prop}
\begin{proof}
If we consider the local mappings for both Rules 160 and 204 we see that the set of blocks on which the rules output differ is $A_0 = \{\begin{smallmatrix}0\\1&0\end{smallmatrix},\begin{smallmatrix}0\\1&1\end{smallmatrix}, \begin{smallmatrix}1\\0&1\end{smallmatrix}, \begin{smallmatrix}1\\1&0\end{smallmatrix}\}$
To use Proposition \ref{prop:distance}, we must find the set $A_n$ in general and thus we must know the sequence of preimages for these particular four basic blocks. We found the first five terms of these sequences and conjectured patterns are $B\left(2^{n+2}-3\right), B, B\left(2^{n+1}-1\right)$ and $B$, respectively, where $B=2^{(n^2+n)/2}$.
From equation (\ref{eqn:distancen}), we determine that
\ben
d \left( \g_{160}^{n+1}, \g_{160}^{n} \circ \g_{204} \right) = 3 \cdot 2^{-n-2} - 4^{-n-1}.
\een
Therefore, since the limit of this expression goes to 0, 
we conclude that Rule 160 emulates identity asymptotically.
\end{proof}

\begin{table}
\centering \caption[Identity Emulation]{Asymptotic Emulation}\label{table:identityemulation}
\begin{tabular}{|| c | c | c ||}
 \hline \hline
Rule $f$ & $d(f^{n+1}, f_{204} \circ f^{n})$ & $P^{(s)}(1)$ \\  \hline \hline
8 & $3 \cdot 2^{-2n-3}$ & 0  \\ \hline
32 & $5 \cdot 2^{-2n-3}$ & 0 \\ \hline
40 & $2^{-n-1}$ & 0 \\ \hline
72 &$4^{-n-2}$ & 0 \\ \hline
128 & $ 2^{(-n^2-3n-2)/2}-  2^{(-n^2-5n-6)/2}$ & 0   \\ \hline
132 & $ 2^{(-n^2-3n-4)/2} $ & $\simeq 0.179$ \\ \hline
136 & $ 2^{-n-2}$ & 0 \\ \hline
140 & $2^{-2n-3}$ & $1/3$ \\ \hline
160 & $3 \cdot 2^{-n-2} - 4^{-n-1} $ & 0  \\ \hline \hline 
\multicolumn{1}{||r|}{} & \multicolumn{1}{r|}{} & \multicolumn{1}{r||}{} \\
Rule $f$ & $d(f^{n+1}, f_{170} \circ f^{n})$ & $P^{(s)}(1)$ \\  \hline \hline
130 & $ 2^{(-n^2-3n-4)/2} $ & $ \simeq 0.179 $ \\ \hline
138 & $ 2^{-2n-3}$ & $1/3$ \\ \hline
162 & $ 2^{-2n-3}$ & $1/3$ \\ \hline\hline
\end{tabular}
\end{table} 

Table \ref{table:identityemulation} shows all known results of rules emulating shift or identity. We can now state our observation
expressed at the beginning of this section using the concept of emulation:\textit{ all rules included in  Tables \ref{table:constant} and \ref{table:decay} emulate identity or shift either in a final number of steps or
asymptotically.}

\section{Further Results: Basic Blocks}
%
We also note that if $\rho=1/2$, it is often possible to compute the number of preimages of other blocks.
 For example, for 40 of the 88 minimal rules, we were able to find preimage sequences for all eight basic blocks, that is, blocks in $\T_2$. In each case, it is only necessary to determine preimage sequences for 5 of the 8 blocks, then we may use Kolmogorov consistency conditions  \cite{dynkin69} to determine the remaining three. In some cases, these formulas are rather striking, such as 
in the case of rule 130, reported in detail in \cite{skelton10}, or rule 172, for which we make the following conjecture.
\begin{conjecture} 
Under 2D CA Rule 172 the preimage sequences of basic blocks are given by
\ben
\card\left[ \g^{-n}(b) \right] = \begin{cases}  2^{(n^2+5n)/2} \sum_{k=0}^{n}
\limits \frac{C_k}{4^k} & \text{if } b \in B_1, \\ 2^{(n^2+5n)/2} \left(2-\sum_{k=0}^{n}\limits \frac{C_k}{4^k}\right) & \text{if } b \in B_2, \end{cases}
\een
where $C_k$ denotes the $k$-th Catalan number and
\begin{align*}
B_1 &= \left\{\begin{smallmatrix} 0 \\ 0 & 0 \end{smallmatrix}, \begin{smallmatrix} 0 \\ 1 & 1 \end{smallmatrix}, \begin{smallmatrix} 1 \\ 0 & 0 \end{smallmatrix}, \begin{smallmatrix} 1 \\ 1 & 1 \end{smallmatrix}\right\}, \quad B_2 = \left\{\begin{smallmatrix} 0 \\ 0 & 1 \end{smallmatrix}, \begin{smallmatrix} 0 \\ 1 & 0 \end{smallmatrix}, \begin{smallmatrix} 1 \\ 0 & 1 \end{smallmatrix}, \begin{smallmatrix} 1 \\ 1 & 0 \end{smallmatrix}\right\}.
\end{align*}
\end{conjecture}
Work on a proof of this result is ongoing and will be reported elsewhere.

\section{Conclusions and future work}
We demonstrated that response curves are calculable for simple rules that  emulate shift or identity. Response curves clearly deserve further study and it is worthwhile to systematically study them for other CA rules.
 However, due to rapidly increasing preimage size, this won't be an easy task for larger neighbourhoods.
One would need a more efficient way to construct the set of preimages of a given block, as simple brute force search becomes  computationally too expensive. We also hope that rigorous results can be obtained for rules with somewhat more complicated dynamics.

\vspace{2mm}
\noindent \textbf{Acknowledgements:}
  One of the authors (HF) acknowledges financial support from the Natural Sciences and Engineering Research Council of Canada (NSERC) in the form of a Discovery Grant. We also wish to thank Shared Hierarchical Academic Research Computing Network (SHARCNET).
  


\end{document}